\documentclass[submission,copyright,creativecommons]{eptcs}

\providecommand{\event}{FROM 2023} 

\usepackage{parskip}
\usepackage{iftex}
\usepackage[utf8]{inputenc}

\ifpdf
  \usepackage{underscore}         
  \usepackage[T1]{fontenc}        
\else
  \usepackage{breakurl}           
\fi

\usepackage{tikz}

\usepackage{enumitem}

\usepackage{amssymb}

\usepackage{amsmath}
\usepackage{amsthm}

\newtheorem{theorem}{Theorem}
\newtheorem{definition}{Definition}
\newtheorem{example}{Example}
\newtheorem{fact}{Remark}
\newtheorem{lemma}{Lemma}
\newcommand{\guard}{\mathit{\beta}}
\newcommand{\opt}[1]{#1?}

\newcommand{\transitionG}[5]{#2\xrightarrow[\ #3\ \rightarrow \ #5\ ]{#1} #4}

\newcommand{\initC}{\mathit{init}}
\newcommand{\emitC}{\mathit{emit}}
\newcommand{\jumpC}{\mathit{jump}}
\newcommand{\statusC}{\mathit{status}} 
\newcommand{\computeC}{\mathit{compute}}

\newcommand{\flattenC}{\mathit{flatten}} 
\newcommand{\extraC}{\mathit{extra}} 
\newcommand{\ObsC}{\mathit{Obs}}

\newcommand{\receiveC}{\mathit{receive}}

\newcommand{\Obs}{\mathit{Obs}}

\newcommand{\msg}{\mathit{msg}}
\newcommand{\status}{\mathit{status}}

\newcommand{\messageE}[3]{\langle#1,\, #2,\, #3\rangle}

\newcommand{\stateCC}[3]{\langle #1, #2, #3\rangle}

\newcommand\mydots{\hbox to 1em{.\hss.\hss.}}
\newcommand{\concat}{\, +\hspace{-.1cm}+\, }

\newcommand{\nomessage}{     
\tikz[baseline]{
\draw[thin,shift={(0 cm,0.06cm)}] circle (1.5pt); 
\draw[thin,shift={(0 cm,0.06cm)}] (-0.08,0.08) -- (0.09,-0.08); 
\draw[thin,shift={(0 cm,0.06cm)}] (-0.08,-0.08) -- (0.09,0.08) }
}

\newcommand{\itspacem}{\vspace{.1cm}}

\title{Asynchronous Muddy Children Puzzle \\ {\Large (work in progress)}}

\author{
\begin{tabular}{cccc}
Dafina Trufaș & Ioan Teodorescu  & Denisa Diaconescu & Traian Șerbănuță  \\
\multicolumn{4}{c}{\footnotesize University of Bucharest} \\ 
\multicolumn{4}{c}{\footnotesize Runtime Verification, Inc.} \\ \\
\multicolumn{4}{c}{Vlad Zamfir} \\
\multicolumn{4}{c}{\footnotesize Independent Researcher} \\ 
\end{tabular}
}
\def\titlerunning{Muddy Children Puzzle in VLSM}
\def\authorrunning{V. Zamfir et al.}
\begin{document}
\maketitle

\begin{abstract}
In this work-in-progress paper we explore using the recently introduced VLSM formalism to define and reason about the dynamics of agent-based systems. To this aim we use VLSMs to formally present several possible approaches to modeling the interactions in the Muddy Children Puzzle as protocols that reach consensus asynchronously.
\end{abstract}

\section{Introduction}

Formally modeling and reasoning about distributed systems with faults is a challenging task~\cite{Fonseca2017}. We recently proposed the theory of \emph{Validating Labeled State transition and Message production systems (VLSMs)} \cite{vlsm-arxiv} as a general approach to modeling and verifying distributed protocols executing in the presence of faults. 

%
%

The theory of VLSMs  has its roots in the work on verification of the CBC Casper protocol~\cite{vlad-2019,vlad-2020} and follows the correct-by-construction methodology for design and development.
Even though the theory of VLSMs was primarily designed for applications in faulty distributed systems, and in blockchains in particular, the framework is general and flexible enough to capture various types of problems in distributed systems. As an illustration, in this paper we show how we can use VLSMs to model and solve (an asynchronous variant of) the epistemic Muddy Children Puzzle~\cite{fagin2004reasoning}.

\section{Classical Muddy Children Puzzle}
Let us begin by recalling the statement of the Muddy Children Puzzle and a classical epistemic logic approach to solving it~\cite{fagin2004reasoning}.

There are $n$ children playing together. It happens during their play that $k$ of the children get mud on their foreheads. Their father comes and says: "At least one of you has mud on your forehead." (if $k > 1$, thus expresses a fact known to each of them before he spoke). The father then asks the following question, repeatedly: "Does any of you know whether you have mud on your own forehead?".

The initial assumptions are expressed in terms of common knowledge. Hence, we shall assume that it is common knowledge that the father is truthful, that all the children can and do hear the father, that all the children can and do see which of the other children besides themselves have muddy foreheads and that all the children are truthful and intelligent enough to make all the necessary deductions during the game.

\paragraph{The solution}
Let's consider first the situation before the father speaks.
We model the problem by a Kripke structure $M=(S, \models, \mathcal{K}_1,\ldots,\mathcal{K}_n)$ over $\Phi$, where:
\begin{itemize}
    \item $\Phi=\{p_i,\ldots,p_n,p\}$ ($p_i$ = "child $i$ has a muddy forehead" and $p$ = "at least one child has a muddy forehead"). Note that $p$ could be obtained as the disjunction of all $p_i$'s; however, for simplicity one can consider it a primitive (albeit non-atomic) predicate.
    \item  $S$ is a set of $2^n$ states (corresponding to all the possible configurations of clean and muddy children, represented as binary tuples)
    \item $(M, (x_1,\ldots,x_n))\models p_i$ iff $x_i=1$\\$(M, (x_1,\ldots,x_n))\models p$ iff $x_i=1$ for some $i$
    \item $(s,t)\in\mathcal{K}_i$ if s and t are identical in all components except eventually the $i^{th}$ one.
\end{itemize}
It's crucial to remark that, in the absence of the father's initial announcement, the fact that "there is at least one muddy child" is not common knowledge and the state of knowledge never changes, no matter how many rounds we take into account. Indeed, after the first question, all the children will certainly answer "No", since they all consider possible the situation in which they themselves do not have mud on their forehead. No information is gained from this round and the situation remains the same after each of the following ones, because each child considers possible a state in which they are clean.

Now let's analyze how the epistemic context changes after the father speaks: as mentioned above, the common knowledge is now larger (even though in the case with $k\geq 2$ muddy children, $p$ was already common knowledge), because of the {\em public} nature of the announcement. Let's consider how the children reason after all of them answered "No" in the first round: it is obvious that all of them eliminate the states containing one muddy child, since the others could not have all answered "No" otherwise. Continuing inductively, we obtain that after $k$ rounds in which all the children answer "No", we can eliminate from the problem graph all the nodes corresponding to states with at most $k$ muddy children. An immediate consequence of this is that after $k-1$ rounds, it becomes common knowledge that there are at least $k$ muddy children. Hence, the muddy children, who each only see $k-1$ muddy children will conclude that they are muddy and answer "Yes".

There are multiple formalizations of this puzzle in the literature~\cite{van2007dynamic,gierasimczuk2011note,Ma2014-MAASA,miedema2023exploiting,gerbrandy1997reasoning}; indeed it seems that each new formalism reasoning about knowledge includes a modeling of this puzzle as a basic example of the expressiveness of the formalism.

\section{VLSMs -- basic notions}

We give a high-level presentation of the theory of VLSMs. More details can be found in \cite{vlsm-arxiv}. 

\begin{definition}[VLSM]
A {Validating Labeled State transition and Message production system} ({\normalfont VLSM}, for short) is a structure of the form
$\mathcal{V} = (L,S,S_0,M, M_0, \tau, \guard),$
where $L$ is a set of \textit{labels}, $(S_0 \subseteq)\ S$ is a non-empty set of \textit{(initial) states}, $(M_0 \subseteq)\ M$ is a set of \textit{(initial) messages}, $\tau : L \times S \times \opt{M} \to S \times \opt{M}$ is a transition function which takes as arguments a label, a state, and possibly a message, and outputs a state and possibly a message, while $\guard$ is a \textit{validity constraint} on the inputs of the transition function.\footnote{\label{note:option-message}For any set $M$ of messages, let $\opt{M} = M \cup \{\nomessage\}$ be the extension of $M$ with $\nomessage$, where $\nomessage$ stands for {\em no-message}.} 
\label{def:vlsm}
\end{definition}

The transition function in a VLSM is total; however, it indirectly becomes a partial function since the validity constraint can filter out some of its inputs. The set of labels in a VLSM can be used to model non-determinism: it is possible to have multiple parallel transitions from a state using the same input message, each with its own label.

\paragraph{Validity.} A transition is called \textit{constrained} if the validity constraint holds on its input. We denote a constrained transition $\tau(l,s,m) = (s',m')$ by
$\transitionG{l}{s}{m}{s'}{m'}.$ 
A \textit{constrained trace} is a sequence of constrained transitions that starts in an initial state. A \textit{valid trace} is inductively defined as a constrained trace in which the input of each transition can be emitted from a valid trace.
A state is \textit{constrained/valid} if there is a constrained/valid trace leading to it. Similarly, a message is \textit{constrained/valid} if it is produced on a constrained/valid trace (we also consider the no-message to be valid). A transition is called \textit{valid} if it is a constrained transition that uses only valid states and messages; thus a valid trace is a sequence of valid transitions starting in an initial state.
%

\paragraph{Equivocation.} In the literature concerning fault-tolerance in distributed systems, equivocation models the fact that certain agents can claim to be in different states to different parties. VLSMs allow modeling such behavior, by specifying that in a valid trace the valid input messages can be produced on different (though still valid) traces.

\paragraph{Composition.} A single VLSM can represent the local point of view of a component in a distributed system. We can obtain the global point of view by composing multiple VLSMs and lifting the local validity constraint of each component. Designers of systems can impose additional restrictions, which are stronger than the ones that can be specified locally on individual components because they can be stated in terms of the global composite state. We capture this phenomenon by the notion of \textit{composition constraint}. 

\begin{definition}[Composition]
Let  $\{\mathcal{V}_i \}_{i=1}^n$ be an indexed set of VLSMs over the same set of messages $M$.
The {constrained composition under a composition constraint $\varphi$}  is the VLSM
$\mathcal{V} =\ \Bigr({\sum_{i=1}^n} \mathcal{V}_i \Bigr) \Bigr|_\varphi = (L,S, S_0, M, M_0,\tau, \guard \wedge \varphi)$
where $L$ is the disjoint union of labels, the (initial) states are the product of (initial) states of the components, the transition function $\tau$  and the constraint predicate $\guard$ are defined component-wise, while the composition constraint $\varphi \subseteq L \times S \times \opt{M}$ is an additional predicate that filters the inputs for the transition function.
\label{def:composition}
\end{definition}
 
When a composition constraint is trivial, i.e., it is the set $L \times S \times \opt{M}$, we refer to the composition as the \textit{free composition} and drop the subscript $\varphi$ in the notation.

\paragraph{No Equivocation constraint.} As mentioned above, VLSMs implicitly allow equivocation. Nevertheless, a truthful behavior of the agents in a composition can still be enforced by means of a {\em no equivocation} constraint, which does not allow receiving in a (composite) state a message from a component if that component could not have emitted the message in a trace leading to the current state.

\section{Asynchronous Muddy Children Puzzle as a VLSM --- take 1}\label{sec:rounds}

We would like to model the Muddy Children Puzzle as a collective effort of a group of agents (the children) to reach a common goal (knowing whether they are muddy or not) by exchanging messages which reveal as little as possible of their personal (initial) knowledge about the problem (a.k.a., which children they see as muddy).  To this aim, there are several characteristics which will be common to both our presented models:

\begin{itemize}
\item Each child needs to know (at all times) which of the other children are muddy (the child's initial knowledge);
\item Every message needs to have a sender;
\item Every message has to report on the epistemic status of its sender at the time the message was sent;
\item Decisions are final: once a child has decided upon their own status (clean/muddy), they will not receive additional messages, as these would not bring new knowledge.
\end{itemize}

In our models we choose to index the children with natural numbers and to represent a child's initial knowledge as a set of indices of the muddy children they see (note that the index of the child cannot appear in this set). We will also use these indices to identify the sender of a message. For the epistemic status, we will use three possible values:
\begin{itemize}
\item $u$ stands for {\em ``child doesn't known their status''};
\item $m$ stands for {\em ``child knows they have mud on their forehead''};
\item $c$ stands for {\em ``child knows they don't have mud on their forehead''}.
\end{itemize}

For our first modeling attempt, we let the children maintain and communicate a number reflecting ``the round number'' (from the original solution) at which they perceive themselves to be.

Let us start formalizing this setting using VLSMs. 
We represent each child as a VLSM of the form $\mathcal{C}_i = (L_i,S_i,S_{0,i},M, M_{0,i}, \tau_i, \guard_i)$. The states of $\mathcal{C}_i$ are either initial states of the form $\langle \Obs \rangle$ where $\Obs$ represents the children seen as muddy by child $i$, or running states of the form $\stateCC{\Obs}{r}{\status}$, where additionally $r$ is the round perceived by child $i$, and $\status$ represents the epistemic status of child $i$.

\begin{itemize}
    \item $L_i = \{\initC, \emitC, \receiveC\}$ reflects the three types of transitions ({\em init\/}ialization, corresponding to the father's announcement and {\em emit\/}ting/{\em receiv\/}ing messages)
    \itspacem  
    \item {$M_i = \{\messageE{j}{r}{\status} \mid j \in \{1,\ldots, n\},\;r\in\mathbb{N},\;\status\in\{u, m, c\}\}$} --- messages also communicate the round number
    \itspacem
    \item $M_{0,i} = \emptyset$
    \itspacem
    \item $S_i = \{\langle \Obs\rangle\ |\ \Obs\subseteq\{1,\mydots, n\}\} \cup \{\langle \Obs,\, s,\, r\rangle\ |\ \Obs\subseteq\{1,\mydots, n\},\ s\in\{u, m, c\},\ r \in \mathbb{N}\}$
    \itspacem
    \item $S_{0,i} = \{\langle \Obs\rangle\ |\ \Obs\subseteq\{1,\mydots, n\}\}$,
\end{itemize}
 
The invariant we would like to maintain is that a child at round $r$ knows (from the messages exchanged with its peers) that there are more than $r$ muddy children.

Assuming that such an invariant holds for all accessible states, then when a child (say $i$) sends a message containing their round number (say $k$) and the fact that they still cannot determine their status, a child receiving such a message (say $j$) can derive from it that $i$ knows that there are more than $k$ muddy children except themselves (since otherwise $i$ would know their status). Moreover, since the receiving child knows whether the sender is muddy or not,  if $j$ sees $i$ as muddy, $j$ can infer that there are actually more than $k+1$ muddy children.
If the current round of $j$ is smaller than the number inferred (either $k$, if $i$ is clean, or $k+1$, if $i$ is muddy), then $j$ can update their current round to that number.

If that happened, say $r$ is the new current round number, if $r$ is less than the number of muddy children that $j$ sees, then the information provided by $r$ is not yet useful enough to draw a conclusion about the status, so the status will stay unknown.
However, if $r$ ever becomes equal to the number of muddy children that $j$ sees, then the child knows that there are more than $r$ muddy children, and since they can only see $r$ muddy children, they will conclude they have to be muddy.

An interesting fact holds for clean children. Note first that they see all the children who are muddy (say $N$), so for them the number of muddy children they see is larger by $1$ than the number of muddy children seen by any muddy child. Hence, they cannot infer their status using the reasoning above, because for that they would have to receive a message from a muddy child, say $i$ with an unknown status at round $N-1$, but that is precisely the number of muddy children that child $i$ sees, so at that round child $i$ would already be able to infer their status.
Let's assume a child $i$ is clean. If there are enough muddy children, $i$ can receive a message from a muddy child with round $N - 2$, and update their round to $N-1$ (using the reasoning above) while maintaining their status as unknown (since the child still sees more muddy children).  Note that $N-1$ is its maximal round according to the invariant proposed above.
But, at the same time, one of the other muddy children, say $j$ can receive the same message and update their round to $N - 1$, which coincides to the number of muddy children they see, so they would change their status to knowing that they are muddy.
Now if $j$ sends a message with their new round and status and $i$ receives it, then $i$ would know they must be clean, but the question is: what would happen with the child's round?
If $i$ keeps the same round (as to not violate the invariant), then they would have two statuses at the same round. To avoid that, we decide to break the invariant in this case and let a clean child advance to round $N$ when they infer they are clean, thus staying in sync with the traditional solution to the puzzle.

Thus we can rephrase the first invariant as:
"a message $\messageE{j}{r}{\status}$ with status different than clean guarantees that $j$ knows that there are more than $r$ muddy children",
and add a second property, saying that a message $\messageE{j}{r}{\status}$ with status different than unknown guarantees that $j$ sees precisely $r$ muddy children.

In the sequel, we propose a transition function (and constraint predicate) to help us realize the proposal above.

\begin{figure}[ht]
    \centering
\begin{align*}
\tau_i(\receiveC, \stateCC{\Obs}{r}{\status}, \messageE{j}{r'}{\status'})
&= \begin{aligned}[t]
     (\stateCC{\Obs}{r}{\status}, \nomessage), & \quad \text{if } \status \in\{m,c\} \\
     (\stateCC{\Obs}{r'}{c}, \nomessage), & \quad \text{if } \status = u \text{ and } \status'=c\\ & \quad \text{and } j \notin \Obs \text{ and } r' = |\Obs| \\
     (\stateCC{\Obs}{r'-1}{m}, \nomessage), & \quad \text{if } \status = u \text{ and } \status'=c\\ & \quad \text{and } j \notin \Obs \text{ and } r' = |\Obs|+1 \\
     (\stateCC{\Obs}{r'}{m}, \nomessage), & \quad \text{if } \status = u \text{ and } \status'=m\\ & \quad \text{and } j \in \Obs \text{ and } r' = |\Obs| \\
     (\stateCC{\Obs}{r'+1}{c}, \nomessage), & \quad \text{if } \status = u \text{ and } \status'=m\\ & \quad \text{and } j \in \Obs \text{ and } r' = |\Obs|-1 \\
     (\stateCC{\Obs}{r}{\status}, \nomessage), & \quad \text{if } \status = u \text{ and } \status' = u\\ & \quad \text{and } j \in \Obs \text{ and } r' < r \\
     (\stateCC{\Obs}{r'+1}{\status}, \nomessage), & \quad \text{if } \status = u \text{ and } \status' = u\\ & \quad \text{and } j \in \Obs \text{ and } r \leq r' < |\Obs| - 1 \\
     (\stateCC{\Obs}{r'+1}{m}, \nomessage), & \quad \text{if } \status = u \text{ and } \status' = u\\ & \quad \text{and } j \in \Obs \text{ and } r' = |\Obs|-1\\
     (\stateCC{\Obs}{r}{\status}, \nomessage), & \quad \text{if } \status = u \text{ and } \status' = u\\ & \quad \text{and } j \notin \Obs \text{ and } r' \leq r \\
     (\stateCC{\Obs}{r'}{\status}, \nomessage), & \quad \text{if } \status = u \text{ and } \status' = u\\ & \quad \text{and } j \notin \Obs \text{ and } r < r' < |\Obs| \\
     (\stateCC{\Obs}{r'}{m}, \nomessage), & \quad \text{if } \status = u \text{ and } \status' = u\\ & \quad \text{and } j \notin \Obs \text{ and } r' = |\Obs| \\
     \stateCC{\Obs}{r}{\status}, \nomessage), & \quad \text{for the cases not treated above}
   \end{aligned}
\end{align*}   
    \caption{The transition function for the $\receiveC$ label}
    \label{fig:enter-label}
\end{figure}

\begin{description}
\item[Init]
From the initial state, each child takes one (silent) transition, analyzing their
current knowledge and initializing the dynamic part of the state accordingly.
\begin{itemize}
    \item The round number is initialized with 0
    \item If the set of muddy children the child sees is empty, then the knowledge flag
      is set to muddy (since at least one must be muddy); otherwise to unknown
\end{itemize}
$\tau_i(\initC, \langle \Obs\rangle, \nomessage) = (\stateCC{\Obs}{0}{u}, \nomessage)$, if $\Obs\neq\emptyset$ \\
$\tau_i(\initC, \langle \Obs\rangle, \nomessage) = (\stateCC{\Obs}{0}{m}, \nomessage)$, if $\Obs=\emptyset$ \\
$\beta_i(\initC, \langle \Obs\rangle, m) = (m = \nomessage)$
\itspacem
\item[Emit]
From any non-initial state, a child can emit a message consisting of their
identifier, current round number and epistemic status,
without changing state.

$\tau_i(\emitC, \stateCC{\Obs}{r}{\status}, \nomessage) = (\stateCC{\Obs}{r}{\status}, \messageE{i}{r}{\status})$ \\
$\beta_i(\emitC, \stateCC{\Obs}{r}{\status}, m) = (m = \nomessage)$

\itspacem
\item[Receive]
To update their state, whenever receiving a message $\langle j, r', \status' \rangle$
in a state $\langle \Obs, r, \status \rangle$, the child does the following:
\begin{itemize}
    \item If their current $\status$ is not $u$, they ignore the message
      (decisions are final).
      
    $\tau_i(\receiveC, \stateCC{\Obs}{r}{\status}, \messageE{j}{r'}{\status'})$ = $(\stateCC{\Obs}{r}{\status}, \nomessage)$, if $\status \in \{m,c\}$
    \item Otherwise:
    \begin{itemize}
        \item If message status ($\status'$) is $c$, and $j$ is not known to be muddy ($j\not\in \Obs$), then from the property above $r'$ must represent the actual number of muddy children. Hence:
        \begin{itemize}
            \item If $r' = |\Obs|$, then child $i$ can update their round to the received message's round and then conclude that they are clean.
            
            $\tau_i(\receiveC, \stateCC{\Obs}{r}{\status}, \messageE{j}{r'}{\status'})$ = $(\stateCC{\Obs}{r'}{c}, \nomessage)$
            \itspacem
            \item If $r' = |\Obs| + 1$, then child $i$ can update their round to the round before received message's round and then conclude that they are muddy.
            
            $\tau_i(\receiveC, \stateCC{\Obs}{r}{\status}, \messageE{j}{r' }{\status'})$ = $(\stateCC{\Obs}{r'-1}{m}, \nomessage)$ \\
        \end{itemize}
        \itspacem
        \item If message status ($\status'$) is $m$ then it must be that $j$ is known as muddy ($j\in \Obs$); then, from the property above, $r'+ 1$ must represent the actual number of muddy children. Hence:
        \begin{itemize}
            \item If $r' = |\Obs|$, then child $i$ can update their round to the received message's round and then conclude that they are muddy.
            
            $\tau_i(\receiveC, \stateCC{\Obs}{r}{\status}, \messageE{j}{r'}{\status'})$ = $(\stateCC{\Obs}{r'}{m}, \nomessage)$
            \itspacem
            \item If $r' = |\Obs| - 1$, then child $i$ can update their round to round after the received message's round and then conclude that they are clean.
            
            $\tau_i(\receiveC, \stateCC{\Obs}{r}{\status}, \messageE{j}{r' }{\status'})$ = $(\stateCC{\Obs}{r'+1}{c}, \nomessage)$ \\
        \end{itemize}
        \itspacem
        \item If message status ($\status'$) is $u$, and $j$ is known as muddy ($j\in \Obs$), then $i$ can infer that $j$ knows that there are more than $r'$ muddy children, and therefore infers that there are more than $r' + 1$ muddy children. Hence:
        \begin{itemize}
            \item If $r' < r$, then child $i$ can ignore the message (it brings nothing new).
            
            $\tau_i(\receiveC, \stateCC{\Obs}{r}{\status}, \messageE{j}{r'}{\status'})$ = $(\stateCC{\Obs}{r}{\status}, \nomessage)$
            \itspacem
            \item If $r \leq r' < |\Obs| - 1$, then child $i$ can update their round to $r' + 1$, but their status will remain unknown (they already know there are at least $|\Obs|$ muddy children).
            
            $\tau_i(\receiveC, \stateCC{\Obs}{r}{\status}, \messageE{j}{r'}{\status'})$ = $(\stateCC{\Obs}{r' + 1}{\status}, \nomessage)$
            \itspacem
            \item If $r' = |\Obs| - 1$, then child $j$ sees at least $|\Obs|$ muddy children. Since the sender is known as muddy, there are at least $|\Obs| + 1$ muddy children. On the other hand, child $i$ knows there are at most $|\Obs| + 1$ muddy children. Combining the two inequalities, we get that there are precisely $|\Obs| + 1$ muddy children, so child $i$ can advance to round $r' + 1$ and knows they are muddy.\\
            $\tau_i(\receiveC, \stateCC{\Obs}{r}{\status}, \messageE{j}{r'}{\status'})$ = $(\stateCC{\Obs}{r'+1}{m}, \nomessage)$ 
        \end{itemize}
        \itspacem
        \item If message status ($\status'$) is unknown, and $j$ is not known as muddy ($j\notin \Obs$), then $i$ can infer that $j$ knows that there are more than $r'$ muddy children, and therefore infers (only) that there are more than $r'$ muddy children. Hence:
        \begin{itemize}
            \item If $r' \leq r$, then child $i$ can ignore the message.
            
            $\tau_i(\receiveC, \stateCC{\Obs}{r}{\status}, \messageE{j}{r'}{\status'})$ = $(\stateCC{\Obs}{r}{\status}, \nomessage)$
            \itspacem
            \item If $r < r' < |\Obs|$, then child $i$ can update its round to $r'$, but its status will remain unknown (it already knows there are at least $|\Obs|$ muddy children).
            
            $\tau_i(\receiveC, \stateCC{\Obs}{r}{\status}, \messageE{j}{r'}{\status'})$ = $(\stateCC{\Obs}{r'}{\status}, \nomessage)$
            \itspacem
            \item If $r' = |\Obs|$, we reason analogous to the similar above case: child $j$ knows that there are at least $|\Obs| + 1$ muddy children and. Adding the fact that child $i$ knows there are at most $|\Obs| + 1$ muddy children, we get that there are precisely $|\Obs| + 1$ muddy children, so child $i$ knows they are muddy.
            
            $\tau_i(\receiveC, \stateCC{\Obs}{r}{\status}, \messageE{j}{r'}{\status'})$ = $(\stateCC{\Obs}{r'}{m}, \nomessage)$ 
        \end{itemize}
    \end{itemize}
\end{itemize}
\end{description}

\paragraph{Composition.}
Note that the notion of VLSM composition in general allows arbitrary initial states for the components. However, in this setting, we need to ensure that the sets of children known by each child to be muddy are {\em consistent}.
Formally, given a child-state $s$, let $\ObsC(s)$ be the observation-set associated to $s$. Then, for any composite state $\sigma = \langle \sigma_1, \ldots, \sigma_n\rangle$, let $\mathbf{consistent}(\sigma)$ be the predicate defined by
\begin{itemize}
    \item $M \neq \emptyset$ (there should be at least one muddy child)
    \item $\ObsC(\sigma_i) = M \setminus \{i\}$, for any $i$ (each child sees all other muddy children)
\end{itemize}
where $M = \bigcup_{i=1}^n{\ObsC(\sigma_i)}$.

Finally, the game flow can be formalized as a constrained VLSM composition.
\begin{center}
$\mathcal{M}uddy\mathcal{P}uzzle = \Bigr(\mathcal{C}_1 + \ldots + \mathcal{C}_n\Bigr) \Bigr|_\varphi$
\end{center} 
where $\varphi$ specifies the following composition constraint:
\begin{description}
\item[init] At the first transition from an initial state we check that the observation sets corresponding to each component are consistent

$\varphi((i, \initC), (\langle \Obs_1\rangle, \ldots, \langle \Obs_n\rangle), m) =  \mathbf{consistent}(\langle \Obs_1\rangle, \ldots, \langle \Obs_n\rangle)$

\item[receive] We must enforce a no-equivocation constraint to ensure the truthfulness of the participants

$\varphi(\langle i, \receiveC\rangle, \langle\sigma_1,\ldots,\sigma_n\rangle, \messageE{j}{r'}{\status'}) = ({\it \status}'=\status_j \wedge r' = r_j) \vee (\status' = u \wedge r' < r_j)$,
where $\sigma_j=\langle \Obs, r_j, \status_j\rangle$.
\end{description}

\subsection{Correctness of the protocol}
In the following, we give a justification of the fact that the above described protocol is correct in the sense that it converges to a solution.
The valid states of the protocol ($S_V$) correspond to the composite states which are VLSM-valid in the constrained composition described above.
We define the set of non-initial valid states $S_V^\ast=S_V\setminus S_0$ and the set of final states $S_F=\{\langle\sigma_i\ldots\sigma_n\rangle\in S_V\mid  \forall i \in \{1,\ldots, n\}, \statusC(\sigma_i)\neq u \}$.

It can be easily checked that the consistency predicate holds for any $\sigma\in S_V^*$:
\begin{fact}
     For each $\sigma\in S_V^*$, $\mathbf{consistent}(\sigma)$ holds.
\end{fact}

Another VLSM-related property, which we state without proof is the fact that the constraint for $\receiveC$ transitions is indeed a no-equivocation constraint:
\begin{fact}[No Equivocation]\label{fct:message}
    For any valid trace $tr$ leading from an initial state $\sigma_0$ to a state $\sigma \in S_v^\ast$ and for any input message $m$ which is valid for $\sigma$ there exists a valid trace starting in $\sigma_0$, of length less than that of $tr$ which emits $m$.
\end{fact}

Let us now show that the invariant that we stated about the dynamics of the protocol is indeed preserved during a (valid) protocol run.

\begin{lemma}[Invariant Preservation]\label{invariant}
For any $\sigma \in S_V$, if $\sigma \not\in S_0$, then:
\item For any component $i$ of $\sigma$ of the form $\sigma_i=\langle \Obs_i, r_i, \status_i\rangle$:
\begin{itemize}
\item If $\status_i=u$, then $r_i<|\Obs_i|$ (and $|\Obs_i| \le N$)
\item If $\status_i=m$, then $r_i=N-1=|\Obs_i|$
\item If $\status_i=c$, then $r_i=N=|\Obs_i|$
\end{itemize}
\end{lemma}
\begin{proof}
Note that it is common knowledge that $|\Obs_i| \le N \le |\Obs_i| + 1$ for any $i$.

We prove the invariant by induction on the length of a valid trace leading to $\sigma$.
The property trivially holds for $\sigma \in S_0$.
For the induction case, we consider the final (valid) transition leading to $\sigma$, say $\transitionG{(i,l)}{\sigma'}{m}{\sigma}{m'}$, and we assume that the invariant holds for $\sigma'$.
We proceed by case analysis on the label of the transition.
\begin{description}
\item[$l = \initC$:] This transition obviously preserves the invariant, because of the father's statement.
\item[$l = \emitC$:] In case of a $\emitC$ transition, the conclusion is also immediate, since the state remains unchanged.
\item[$l = \receiveC$:] From Remark~\ref{fct:message} there is a composite valid state from which the input message can be emitted which has the same observation sets as $\sigma'$ (since they both are reachable from the same initial state) and for which  we can apply the induction hypothesis and thus assume that the invariant holds.
\begin{itemize}
    \item $\tau_i(\receiveC, \stateCC{\Obs_i}{r}{m}, \messageE{j}{r'}{\status'})$: the message is ignored and the state remains unchanged, so the conclusion is immediate.
    \item $\tau_i(\receiveC, \stateCC{\Obs_i}{r}{c}, \messageE{j}{r'}{\status'})$: we proceed analogous to the previous case.
    \item $\tau_i(\receiveC, \stateCC{\Obs_i}{r}{u}, \messageE{j}{r'}{c})$ and $j \notin \Obs_i$ and $r' = |\Obs_i|$:\\
        By applying the induction hypothesis to the state from which the message is obtained, we have $r'=N=|\Obs_j|$.\\
        The resulting state $\stateCC{\Obs_i}{r'}{c}$ preserves the invariant, because $r'=N=|\Obs_i|$.
    \item $\tau_i(\receiveC, \stateCC{\Obs_i}{r}{u}, \messageE{j}{r'}{c})$ and $j \notin \Obs_i$ and $r' = |\Obs_i| + 1$:\\
        By applying the induction hypothesis to the state from which the message is obtained, we have $r'=N=|\Obs_j|$.\\
        The resulting state $\stateCC{\Obs_i}{r'-1}{m}$ preserves the invariant, because $r'-1=N-1=|\Obs_i|$.
    \item $\tau_i(\receiveC, \stateCC{\Obs_i}{r}{u}, \messageE{j}{r'}{m})$ and $j \in \Obs_i$ and $r' = |\Obs_i|$:\\
        By applying the induction hypothesis to the state from which the message is obtained, we have $r'=N-1=|\Obs_j|$.\\
        The resulting state $\stateCC{\Obs_i}{r'}{m}$ preserves the invariant, because $r'=N-1=|\Obs_i|$.
    \item $\tau_i(\receiveC, \stateCC{\Obs_i}{r}{u}, \messageE{j}{r'}{m})$ and $j \in \Obs_i$ and $r' = |\Obs_i| - 1$:\\
        By applying the induction hypothesis to the state from which the message is obtained, we have $r'=N-1=|\Obs_j|$.\\
        The resulting state $\stateCC{\Obs_i}{r'+1}{c}$ preserves the invariant, because $r'+1=N=|\Obs_i|$.
    \item $\tau_i(\receiveC, \stateCC{\Obs_i}{r}{u}, \messageE{j}{r'}{u})$ and $j \in \Obs_i$ and $r' < r$:\\
        The state after applying the transition remains unchanged, so the conclusion is immediate.
    \item $\tau_i(\receiveC, \stateCC{\Obs_i}{r}{u}, \messageE{j}{r'}{u})$ and $j \in \Obs_i$ and $r \leq r' < |\Obs_i|-1$:\\
        The resulting state $\stateCC{\Obs_i}{r'+1}{u}$ obviously preserves the invariant, since $r'+1<|\Obs_i|$.
    \item $\tau_i(\receiveC, \stateCC{\Obs_i}{r}{u}, \messageE{j}{r'}{u})$ and $j \in \Obs_i$ and $r' = |\Obs_i| - 1$:\\
        We have $r'+1=|\Obs_i|\geq N-1$, so $r'\geq N-2$.\\
        On the other hand, applying the induction hypothesis, we have that $r'<  |\Obs_j|$ and since child $j$ is muddy, $|\Obs_j|=N-1$, and combining these we get $r'<|\Obs_j|=N-1$, which implies $r'\leq N-2$.\\
        We can conclude that $r'=N-2$, so the resulting state $\stateCC{\Obs_i}{r'+1}{m}$ preserves the invariant, because $r'+1=|\Obs_i|=N-1$.
    \item $\tau_i(\receiveC, \stateCC{\Obs_i}{r}{u}, \messageE{j}{r'}{u})$ and $j \notin \Obs_i$ and $r' <= r$:\\
        The state after applying the transition remains unchanged, so the conclusion is immediate.
    \item $\tau_i(\receiveC, \stateCC{\Obs_i}{r}{u}, \messageE{j}{r'}{u})$ and $j \notin \Obs$ and $r < r' < |\Obs_i|$:\\
        The resulting state $\stateCC{\Obs_i}{r'}{\status}$ preserves the invariant, since $r' < |\Obs_i|$.
    \item $\tau_i(\receiveC, \stateCC{\Obs_i}{r}{u}, \messageE{j}{r'}{u})$ and $j \notin \Obs_i$ and $r' = |\Obs_i|$:\\
        By applying the induction hypothesis to the state from which the message is obtained, we have $r'<|\Obs_j|$ and since child $j$ is clean, $|\Obs_j|=N$ so we get $r'<N$ and combining this with the last condition in the transition, we immediately obtain that $|\Obs_i|<N$.\\
        But since it is common knowledge that $|\Obs_i|\geq N-1$ it must be that $|\Obs_i| = N-1$.
        We can conclude that $r' = |\Obs_i| = N-1$, so the resulting state $\stateCC{\Obs_i}{r'}{m}$ preserves the invariant.
    \item For all the cases not treated above, the child's $\beta$ predicate does not hold, so there cannot be any transition.
\end{itemize}
\end{description}
\end{proof}

\begin{theorem}
    From any initial consistent state, there is a path leading to a final state in which each child's status is consistent with the instance of the problem.
\end{theorem}
\begin{proof}[Proof sketch]
The result can be obtained from the following properties:
\begin{description}
\item[Progress] From any valid non-final state $\sigma$, there is a valid transition leading to a state $\sigma'$ on some component $i$, such that $round(\sigma'_i) > round(\sigma_i)$, where we let $round(\langle \Obs \rangle)$ be $-1$.

If there are any component initial states left, we advance one of them. If there are any children which already know their status, we advance to the final state any of the others (there must be at least one child with unknown status since the state is not final). If no child knows their status yet, then there must be at least two muddy children. Then take the one among them with minimum round number (if their round numbers are equal, take any of them) and receive the message corresponding to the current state of the other. This will surely increase their current round number.

\item[Invariant Preservation] The invariant from Lemma~\ref{invariant} holds.

In particular, if a final state is reached, each child's status will be consistent with the current instance of the problem.

\item[Termination] No child can increase their round number past the number of muddy children they see.

Easy to prove by induction and analysis on the cases of the transition function.
\end{description}
\end{proof}

\subsection{Possible optimization}

If we analyze the dynamics of the solution proposed above, we notice that the only relevant exchange of messages happens in the final stages of the protocol.
Indeed, assuming a child sees $n$ muddy children, they know that any other child sees at least $n-1$ muddy children; therefore no message it would send with round less than $n - 1$ can help any other child determine their status. Hence, the child can "jump" during the initialization phase directly to the round $n-1$.  Formally, we add a new label \textbf{jump} and replace the \textbf{Init} transition with the following:

\begin{quote}
From any state with status unknown, a child can take one (silent) transition "jumping" to a future state it knows it can reach based on existing knowledge, where the round number becomes the number preceding the number of muddy children the child sees and the child's knowledge flag remains unknown.

$\tau_i(\jumpC, \stateCC{\Obs}{r}{\status}, m) = (\stateCC{\Obs}{|\Obs|-1}{u}, \nomessage)$\\
$\beta_i(\jumpC, \stateCC{\Obs}{r}{\status}, m) = (\status = u) \wedge (m = \nomessage) \wedge (r < |\Obs| - 1)$
\end{quote}

After such a jump one can find a very short path to a solution:
\begin{enumerate}
    \item a child sends a message after the jump with the unknown status
    \item another (muddy) child receives that message and discovers they are muddy and sends a message with this discovery
    \item all other children (including the first) receive this second message and discover their status
\end{enumerate}

\subsection{Discussion}

The above proposed solution seems to satisfy our initial guidelines for a good asynchronous solution.
Nevertheless, it is far from being perfect, as illustrated by the following example. 

\begin{example}[Information leak]\label{ex:issue}
Assume there are $5$ children, $4$ of whom are muddy and one is not. Then a muddy child (say, child $1$) can discover that they are muddy by only exchanging messages with the clean child (say, child $5$) with the following scenario:

\begin{enumerate}
\item $1$ at round $0$ sends message $m^1_0$ that they do not know
\item $5$ at round $0$ receives message $m^1_0$, advances to round $1$, and sends message $m^5_1$ that they do not know
\item $1$ at round $0$ receives the message $m^5_1$, advances to round $1$, and sends message $m^1_1$ that they do not know
\item $5$ at round $1$ receives message $m^1_1$, advances to round $2$, and sends message $m^5_2$ that they do not know
\item $1$ at round $1$ receives the message $m^5_2$, advances to round $2$, and sends message $m^1_2$ that they do not know
\item $5$ at round $2$ receives message $m^1_2$, advances to round $3$, and sends message $m^5_3$ that they do not know
\item $1$ at round $2$ receives the message $m^5_3$, advances to round $3$ which equals the number of muddy children they see and changes their status to knowing they are muddy (and sends message $m^1_3$ with status muddy)
\end{enumerate}
Moreover, once $1$ knows they are muddy and sends a message about it, upon receiving that message all the other children will know their status.
\end{example}

There are at least two issues about the above example: (1) that a muddy child needs no information from other muddy children to discover that they are muddy; and (2) that once a child knows their status, all other children immediately know their status, even if they have not taken part in the conversation so far.

It thus seems that, although sharing the perceived round as part of the message helps with making the discovery process asynchronous, it also alters it more than expected in terms of what becomes inferable during an exchange of information. The following  section proposes a solution we believe to be closer to the original formulation and intended dynamics, while staying asynchronous.

\section{Solving the puzzle by extracting information from messages}

To alleviate the issues identified with the solution in Section~\ref{sec:rounds}, we propose a new solution, which follows more closely an epistemic logic point of view, in the sense that the messages exchanged between children only carry information that the child has previously seen. This guarantees that deduction can only be done according to information that is (was) publicly available.

To do that, we define the child $i$ as the following VLSM, resembling the previous encoding, but using {\em message histories} (lists of messages previously received) instead of round numbers:

\begin{itemize}
    \item $L_i = \{\initC, \emitC, \receiveC\}$
    \itspacem  
    \item $M = \{\langle j,\, s,\, h\rangle\ |\ j \in \{1,\mydots,n\},\ s \in \{u,m,c\},\ h \in M^*\}$
    \itspacem
    \item $M_{0,i} = \emptyset$
    \itspacem
    \item $S_i = \{\langle \Obs\rangle\ |\ \Obs\subseteq\{1,\mydots, n\}\} \cup \{\langle \Obs,\, s,\, h\rangle\ |\ \Obs\subseteq\{1,\mydots, n\},\ s\in\{u, m, c\},\ h \in M^*\}$
    \itspacem
    \item $S_{0,i} = \{\langle \Obs\rangle\ |\ \Obs\subseteq\{1,\mydots, n\}\}$,
\end{itemize}

Let us now describe the dynamics of a child's interaction with the others.
\begin{description}
\item[Init] The initialization phase will remain basically the same as for the previous solution, except that, instead of round $0$, now the children will have the empty history:
    
$\tau_i(\initC,\langle \Obs \rangle, \nomessage) = 
\left \{
\begin{array}{rl}
(\langle \Obs,\, u,\, [] \rangle, \nomessage), & \mbox{ if } \Obs \neq \emptyset \\
(\langle \Obs,\, m,\, [] \rangle, \nomessage), & \mbox{ if } \Obs = \emptyset \\
\end{array}
\right.
$ \\
$\beta_i(\initC, \langle \Obs\rangle, \msg) = (\msg = \nomessage)$

\item[Emit] Similarly, the emit phase is basically the same, except that, instead of sending their round number, the children will  now send their current history:
 
$\tau_i(\emitC,\langle \Obs,\, s,\, h \rangle, \nomessage) = (\langle \Obs, s, h \rangle,\ \langle i,\, s,\, h \rangle)$ \\
$\beta_i(\emitC,\langle \Obs,\, s,\, h \rangle, \msg) = (\msg = \nomessage)$

\item[Receive] Upon receiving a message (in a not yet known status) the child will aggregate the message received with the 
information they already had (append it to its history) and compute its new status based on its $\Obs$ervation set and its new history.

$\tau_i(\receiveC,\langle \Obs,\, s,\, h \rangle, \langle j,\, s_j,\, h_j \rangle) = $ 

\hspace{1cm}
$
\left \{
\begin{array}{rl}
(\langle \Obs,\, s,\, h \rangle, \nomessage), & \mbox{ if } s \in \{c,m\} \\
(\langle \Obs,\, \computeC(\Obs, h \concat [\langle j, s_j, h_j \rangle]),\, h \concat [\langle j,\, s_j,\, h_j \rangle] \rangle, \nomessage), & \mbox{ if } s = u \\
\end{array}
\right.
$

\textit{// we compute the new status} \\
$\computeC(\Obs, h) = 
\left \{
\begin{array}{rl}
m, & \mbox{ if } \min_{k\in \Obs} |\mathit{groupSimilar}_i(\mathit{unknown}_k(\flattenC(h)))| \geq |\Obs| \\
c, & \mbox{ if the first rule didn't apply and } |\mathit{muddy}(\flattenC(h))| = |\Obs| \\
u, & \mbox{ otherwise}
\end{array}
\right.
$\\

\textit{// we unfold all the messages contained (in depth) in the history} \\
$\flattenC(h \concat [\langle j,\, s_j,\, h_j \rangle]) = 
\flattenC(h)\ \cup\ \flattenC(h_j)\ \cup\ \{\langle j,\, s_j,\, h_j \rangle\}\ \cup\  \extraC(\langle j,\, s_j,\, h_j \rangle)$ \\ 

\textit{// we extract extra information from a message about messages which could have been emitted by its sender (e.g. in all prefixes of the history in which the sender didn't know their status)} \\ 
$\extraC(\langle j,\, s_j,\, h_j \rangle) = \{\langle j,u,h_j'\rangle \ |\ h_j'\  \triangleleft\ h_j\}$\footnote{We denote by $h_1 \triangleleft h_2$ the fact that $h_1$ is a strict prefix of $h_2$.} \\ 

\textit{// we select messages belonging to a given child and having an unknown status} \\ 
$\mathit{unknown}_k(h) = \{m \mid m \in h \wedge m = \langle k,u,h'\rangle \mbox{ for some h'}\}$ \\

\textit{// we group messages that carry the same information relative to child $i$} \\
$\mathit{groupSimilar}_i(h) = h \setminus \{\langle j,\ u ,\ h' \concat [\langle i,\_,\_ \rangle]\rangle \mbox{ for some h'}\}$  \\

\textit{// we construct the set of all children who know they are muddy in a history} \\ 
$\mathit{muddy}(h) = \{j \mid \langle j,m,h_j\rangle \in h\}$
\end{description}

Finally, the game can be formalized as a constrained VLSM composition.
\begin{center}
$\mathcal{M}uddy\mathcal{P}uzzle = \Bigr(\mathcal{C}_1 + \mathcal{C}_2 + \mathcal{C}_3 \Bigr) \Bigr|_\varphi$
\end{center} 
where

\textit{// consistency}  \\
$\varphi((i, \initC), (\langle \Obs_1\rangle, \langle \Obs_2\rangle, \langle \Obs_3\rangle), m) = \mathbf{consistent}(\langle \Obs_1\rangle, \langle \Obs_2\rangle, \langle \Obs_3\rangle)$,\\ \\

\textit{// no equivocation} \\
$\varphi((i, \receiveC), \langle \sigma_1,\sigma_2,\sigma_3\rangle, \langle j, s_j,h_j \rangle) =  $

\hspace{1cm}$(s_j = \statusC(\sigma_j)\ \wedge\ h_j = \mathit{history}(\sigma_j))\ \vee\  (s_j = u\ \wedge\ h_j \triangleleft\ \mathit{history}(\sigma_j)),$

where for a state $\sigma = \langle \Obs, s , h\rangle$, we define $\statusC(\sigma) = s$ and $\mathit{history}(\sigma) =  h$.

It is relatively easy to see that the kind of reasoning described in the (counter-)Example~\ref{ex:issue} is no longer possible, because after the first exchange of messages, child $1$ cannot really infer anything, because all they see is a message from child $5$ containing a message from themselves.

We argue that the information contained in a message represents the epistemic knowledge of its sender about the status of the other children at the time the message was sent, thus making this solution much closer to the standard solution (using synchronous rounds and public announcements).

First, let $N$ be the total number of muddy children, and let us observe that if a child knows $N$, then they also know their status (by comparing $N$ with the number of muddy visible children). On the other hand, telling the others that they know $N$ (without actually telling them the value) is no different than telling them that their status can be inferred.
Let $K_1, \ldots, K_n$ be the epistemic operators assigned to each child, with the intuitive meaning of $K_i \psi$ that $i$ can infer that $\psi$ holds, based on the commonly available information, and on the private information that $i$ has.  Let also $q_1, \ldots, q_n$ be primitive propositions, $q_i$ meaning that $i$ knows the value of $N$.

Then, we can recursively encode a message $m$ as a formula $E(m)$, as follows:

$$
E(\langle j, s_j, h_j \rangle) = \left\{
 \begin{array}{lc}
 (K_j \bigwedge_{m\in h_j} E(m)) \wedge K_j\left(\bigwedge_{m\in h_j} E(m) \rightarrow q_j\right) & \mbox{ if } s_j \neq u \\
   (K_j \bigwedge_{m\in h_j} E(m)) \wedge \neg K_j\left(\bigwedge_{m\in h_j} E(m) \rightarrow q_j\right) & \mbox{ if } s_j = u \\
 \end{array}
\right. 
$$

One important observation is that the formulas corresponding to all additional messages introduced by {\em flatten(h)} are directly deducible from the formulas of $h$.

\paragraph{Discussion.}  Our preliminary exploration of the interactions possible within the model seems to hint that a solution is reachable from any initial state for the case of three children. Nevertheless, our current definition of $\textit{groupSimilar}$ was engineered for the case of three children and we know it doesn't scale to more children as-is.  We believe that for the general case it should be replaced by factoring the messages through an equivalence relation grouping messages which carry the similar amount of information as perceived by the child receiving these messages.

\section{Conclusion and Future Work}

In this paper we have shown that compositions of VLSMs seem like a useful tool to model the dynamics of distributed agents, allowing independent description of their behavior, but also enabling the specification of global constraints (e.g., truthfulness, fairness, etc.).

Regarding the models proposed for the Muddy Children Puzzle, for the first model we were able to show that it achieves the goal of converging towards a solution, but also showed that it might leak more information than intended through communication. We believe the second model to be promising and more faithful to our modeling goals and plan to further explore its possible generalizations to an arbitrary number of children.


\begin{thebibliography}{10}
\providecommand{\bibitemdeclare}[2]{}
\providecommand{\surnamestart}{}
\providecommand{\surnameend}{}
\providecommand{\urlprefix}{Available at }
\providecommand{\url}[1]{\texttt{#1}}
\providecommand{\href}[2]{\texttt{#2}}
\providecommand{\urlalt}[2]{\href{#1}{#2}}
\providecommand{\doi}[1]{doi:\urlalt{https://doi.org/#1}{#1}}
\providecommand{\eprint}[1]{arXiv:\urlalt{https://arxiv.org/abs/#1}{#1}}
\providecommand{\bibinfo}[2]{#2}

\bibitemdeclare{book}{fagin2004reasoning}
\bibitem{fagin2004reasoning}
\bibinfo{author}{R.~\surnamestart Fagin\surnameend}, \bibinfo{author}{J.Y.
  \surnamestart Halpern\surnameend}, \bibinfo{author}{Y.~\surnamestart
  Moses\surnameend} \& \bibinfo{author}{M.~\surnamestart Vardi\surnameend}
  (\bibinfo{year}{2004}): \emph{\bibinfo{title}{Reasoning About Knowledge}}.
\newblock \bibinfo{series}{A Bradford Book}, \bibinfo{publisher}{MIT Press},
  \doi{10.7551/mitpress/5803.001.0001}.

\bibitemdeclare{inproceedings}{Fonseca2017}
\bibitem{Fonseca2017}
\bibinfo{author}{Pedro \surnamestart Fonseca\surnameend},
  \bibinfo{author}{Kaiyuan \surnamestart Zhang\surnameend},
  \bibinfo{author}{Xi~\surnamestart Wang\surnameend} \& \bibinfo{author}{Arvind
  \surnamestart Krishnamurthy\surnameend} (\bibinfo{year}{2017}):
  \emph{\bibinfo{title}{An Empirical Study on the Correctness of Formally
  Verified Distributed Systems}}.
\newblock In: {\slshape \bibinfo{booktitle}{European Conference on Computer
  Systems}}, pp. \bibinfo{pages}{328--343}, \doi{10.1145/3064176.3064183}.

\bibitemdeclare{article}{gerbrandy1997reasoning}
\bibitem{gerbrandy1997reasoning}
\bibinfo{author}{Jelle \surnamestart Gerbrandy\surnameend} \&
  \bibinfo{author}{Willem \surnamestart Groeneveld\surnameend}
  (\bibinfo{year}{1997}): \emph{\bibinfo{title}{Reasoning about information
  change}}.
\newblock {\slshape \bibinfo{journal}{Journal of logic, language and
  information}} \bibinfo{volume}{6}, pp. \bibinfo{pages}{147--169},
  \doi{10.1023/A:1008222603071}.

\bibitemdeclare{inproceedings}{gierasimczuk2011note}
\bibitem{gierasimczuk2011note}
\bibinfo{author}{Nina \surnamestart Gierasimczuk\surnameend} \&
  \bibinfo{author}{Jakub \surnamestart Szymanik\surnameend}
  (\bibinfo{year}{2011}): \emph{\bibinfo{title}{A note on a generalization of
  the muddy children puzzle}}.
\newblock In: {\slshape \bibinfo{booktitle}{Proceedings of the 13th Conference
  on Theoretical Aspects of Rationality and Knowledge}}, pp.
  \bibinfo{pages}{257--264}, \doi{10.1145/2000378.2000409}.

\bibitemdeclare{conference}{vlad-2020}
\bibitem{vlad-2020}
\bibinfo{author}{Elaine \surnamestart Li\surnameend},
  \bibinfo{author}{Traian~Florin \surnamestart Șerbănuță\surnameend},
  \bibinfo{author}{Denisa \surnamestart Diaconescu\surnameend},
  \bibinfo{author}{Vlad \surnamestart Zamfir\surnameend} \&
  \bibinfo{author}{Grigore~Ro\c \surnamestart su\surnameend}
  (\bibinfo{year}{2020}): \emph{\bibinfo{title}{Formalizing
  {Correct-By-Construction Casper} in {Coq}}}.
\newblock In: {\slshape \bibinfo{booktitle}{International Conference on
  Blockchain and Cryptocurrency}}, pp. \bibinfo{pages}{1--3},
  \doi{10.1109/ICBC48266.2020.9169468}.

\bibitemdeclare{article}{Ma2014-MAASA}
\bibitem{Ma2014-MAASA}
\bibinfo{author}{Minghui \surnamestart Ma\surnameend},
  \bibinfo{author}{Alessandra \surnamestart Palmigiano\surnameend} \&
  \bibinfo{author}{Mehrnoosh \surnamestart Sadrzadeh\surnameend}
  (\bibinfo{year}{2014}): \emph{\bibinfo{title}{Algebraic Semantics and Model
  Completeness for Intuitionistic Public Announcement Logic}}.
\newblock {\slshape \bibinfo{journal}{Annals of Pure and Applied Logic}}
  \bibinfo{volume}{165}(\bibinfo{number}{4}), pp. \bibinfo{pages}{963--995},
  \doi{10.1016/j.apal.2013.11.004}.

\bibitemdeclare{article}{miedema2023exploiting}
\bibitem{miedema2023exploiting}
\bibinfo{author}{Daniel \surnamestart Miedema\surnameend} \&
  \bibinfo{author}{Malvin \surnamestart Gattinger\surnameend}
  (\bibinfo{year}{2023}): \emph{\bibinfo{title}{Exploiting Asymmetry in Logic
  Puzzles: Using ZDDs for Symbolic Model Checking Dynamic Epistemic Logic}}.
\newblock {\slshape \bibinfo{journal}{arXiv preprint arXiv:2307.05067}},
  \doi{10.48550/arXiv.2307.05067}.

\bibitemdeclare{book}{van2007dynamic}
\bibitem{van2007dynamic}
\bibinfo{author}{Hans \surnamestart Van~Ditmarsch\surnameend},
  \bibinfo{author}{Wiebe \surnamestart van Der~Hoek\surnameend} \&
  \bibinfo{author}{Barteld \surnamestart Kooi\surnameend}
  (\bibinfo{year}{2007}): \emph{\bibinfo{title}{Dynamic epistemic logic}}.
\newblock \bibinfo{volume}{337}, \bibinfo{publisher}{Springer Science \&
  Business Media}, \doi{10.1007/978-1-4020-5839-4}.

\bibitemdeclare{article}{vlsm-arxiv}
\bibitem{vlsm-arxiv}
\bibinfo{author}{V.~\surnamestart Zamfir\surnameend},
  \bibinfo{author}{M.~\surnamestart Calancea\surnameend},
  \bibinfo{author}{D.~\surnamestart Diaconescu\surnameend},
  \bibinfo{author}{W.~Ko\l \surnamestart owski\surnameend},
  \bibinfo{author}{B.~\surnamestart Moore\surnameend},
  \bibinfo{author}{K.~\surnamestart Palmskog\surnameend}, \bibinfo{author}{T.F.
  \surnamestart Șerbănuță\surnameend}, \bibinfo{author}{M.~\surnamestart
  Stay\surnameend}, \bibinfo{author}{D.~\surnamestart Trufaș\surnameend} \&
  \bibinfo{author}{J.~\surnamestart {Tu\v sil}\surnameend}
  (\bibinfo{year}{2022}): \emph{\bibinfo{title}{Validating Labelled State
  Transition and Message Production Systems: A Theory for Modelling Faulty
  Distributed Systems}}.
\newblock {\slshape \bibinfo{journal}{arXiv:2202.12662v3}},
  \doi{10.48550/ARXIV.2202.12662}.

\bibitemdeclare{unpublished}{vlad-2019}
\bibitem{vlad-2019}
\bibinfo{author}{Vlad \surnamestart Zamfir\surnameend}, \bibinfo{author}{Nate
  \surnamestart Rush\surnameend}, \bibinfo{author}{Aditya \surnamestart
  Asgaonkar\surnameend} \& \bibinfo{author}{Georgios \surnamestart
  Piliouras\surnameend} (\bibinfo{year}{2019}):
  \emph{\bibinfo{title}{Introducing the ``Minimal {CBC Casper}'' Family of
  Consensus Protocols}}.
\newblock \bibinfo{note}{\url{https://github.com/cbc-casper/cbc-casper-paper}}.

\end{thebibliography}

\end{document}